\documentclass[aps,showpacs,pra,twocolumn,superscriptaddress]{revtex4-2}

\usepackage{exscale,amsfonts,mathrsfs}
\usepackage{bbm}
\usepackage{graphicx}
\usepackage{latexsym}
\usepackage{times}
\usepackage{subfigure}
\usepackage[T1]{fontenc}
\usepackage{enumerate}
\usepackage{bbold}
\usepackage{color}
\usepackage{mathtools}
\usepackage[normalem]{ulem}
\usepackage{amsfonts,amsmath,amssymb,amsthm,dsfont}
\usepackage[colorlinks=true,citecolor=blue,urlcolor=blue]{hyperref}
\usepackage{xcolor}
\usepackage{natbib}
\usepackage{changes}
\usepackage{float}

\theoremstyle{plain}
\newtheorem{Lemma}{Lemma}
\newtheorem{Theorem}{Theorem}

\newcommand{\1}{\mathbb{1}}
\newcommand{\W}{\mathcal{W}}
\newcommand{\Wo}[1]{\W^{(#1)}_{\mathrm{opt}}}
\newcommand{\Wt}{\widetilde{\W}}
\newcommand{\p}{p_{\textrm{limit}}}
\newcommand{\TrM}[1]{\text{Tr}\left(#1\right)}
\newcommand{\Vs}{V_{\mathbb{S}}}
\newcommand{\ket}[1]{|#1\rangle}
\newcommand{\bra}[1]{\langle #1|}
\newcommand{\Tr}{\mathrm{Tr}}
\newcommand{\proj}[1]{|#1\rangle\!\langle#1|}
\newcommand{\scalprod}[2]{\langle #1|#2\rangle}

\providecommand{\DIFdelbegin}{}
\providecommand{\DIFdelend}{}
\providecommand{\DIFaddbegin}{}
\providecommand{\DIFaddend}{}

\begin{abstract}
Genuine multipartite entanglement is arguably the most valuable form of entanglement in the multipartite case with application for instance in quantum metrology. In order to detect that form of entanglement in multipartite quantum states, one typically uses entanglement witnesses. The aim of this paper is to generalize the results of [G. T\'oth and O. G\"uhne, 
Phys. Rev. A \textbf{72}, 022340 (2005)] in order to provide a construction of witnesses of genuine multipartite entanglement tailored to entangled subspaces originating from the \textit{multi-qudit} stabilizer formalism---a framework well known for its role in quantum error correction, which also provides a very convenient description of a broad class of entangled multipartite states (both pure and mixed). Our construction includes graph states of arbitrary local dimension. We then show that in certain situations, the obtained witnesses detecting genuine multipartite entanglement in quantum systems of higher local dimension are superior in terms of noise robustness to those derived for multiqubit states.

\end{abstract}

\begin{document}
\title{Entanglement witnesses for stabilizer states and subspaces beyond qubits}
\author{Jakub Szczepaniak}
\affiliation{Center for Theoretical Physics, Polish Academy of Sciences, Aleja Lotnik\'{o}w 32/46, 02-668 Warsaw, Poland}
\author{Owidiusz Makuta}
\affiliation{Center for Theoretical Physics, Polish Academy of Sciences, Aleja Lotnik\'{o}w 32/46, 02-668 Warsaw, Poland}

\affiliation{Instituut-Lorentz, Universiteit Leiden, P.O. Box 9506, 2300 RA Leiden, The Netherlands}
\affiliation{$\langle \text{aQa}^\text{L} \rangle$ Applied Quantum Algorithms Leiden, The Netherlands}

\author{Remigiusz Augusiak}
\email{Corresponding author (augusiak@cft.edu.pl)}
\affiliation{Center for Theoretical Physics, Polish Academy of Sciences, Aleja Lotnik\'{o}w 32/46, 02-668 Warsaw, Poland}

\maketitle

\section{Introduction}

Entanglement is considered to be one of the most characteristic features of quantum theory that has become a resource for interesting applications such as quantum key distribution \cite{PhysRevLett.67.661} or quantum teleportation \cite{PhysRevLett.70.1895} (see also Ref. \cite{Horodecki_2009} for a review). It is also the key phenomenon giving rise to Bell nonlocality \cite{Bell}, which is another form of nonclassical correlations that quantum theory exhibits. In the multipartite setting, the most valuable form of entanglement is considered to be the so-called genuine multipartite entanglement (GME), which, among others, is vital for the possibility of achieving the Heisenberg limit in quantum metrology \cite{MultipMetrology,MultipMetrology2}.    

However, the full exploitation of entanglement as a resource requires efficient tools that allow its detection in composite quantum states. Arguably, one of the most powerful and at the same time most complete methods is that based on entanglement witnesses (EWs), first introduced for bipartite quantum states in the seminal paper \cite{HORODECKI19961}, and later extended to the multipartite case in \cite{HORODECKI20011}; notice that the term itself was coined only later in Ref. \cite{TERHAL2000319}. Entanglement witnesses are Hermitian observables whose expectation values can be used to detect entanglement in composite quantum states. Specifically, their expectation values are nonnegative for all separable states \cite{PhysRevA.40.4277}, while certain entangled states yield negative expectation values, signaling the presence of entanglement. At the same time, for any entangled state, there exists an entanglement witness which detects entanglement of that state \cite{HORODECKI19961,HORODECKI20011}. Importantly, entanglement witnesses can also be used to detect genuine multipartite entanglement (see, e.g., Refs. \cite{PhysRevLett.94.060501,OtfriedStabilizer}). 

For these reasons, entanglement witnesses have been intensively studied in the literature \cite{PhysRevA.62.052310,PhysRevA.84.052323,doi:10.1142/S0219749915500604} and many constructions of them have been devised to detect entanglement in various bipartite and multipartite quantum states, just to mention Refs. \cite{PhysRevLett.94.060501,OtfriedStabilizer,PhysRevA.80.062314,PhysRevLett.111.110503,PhysRevLett.123.100507,PhysRevA.97.032318} (cf. also the reviews \cite{GUHNE20091,Chruściński_2014}). However, the majority of these efforts have focused on the simplest case of multipartite systems composed of qubits, whereas systems with higher local dimensions have received a lot less attention. At the same time, significant progress has recently been made in the preparation and manipulation of various multiqudit quantum systems (see e.g. Refs. \cite{Wang,HighDim,Martin}), providing strong motivation to develop suitable entanglement witnesses capable of detecting genuine multipartite entanglement in such systems. Despite the fact that a number of entanglement criteria have been designed up to date, the problem is far from being completely solved. 

The main aim of this work is to introduce constructions of witnesses of GME for quantum states originating from the stabilizer formalism 
\cite{Steane,PhysRevA.54.3824,gottesman1997stabilizercodesquantumerror,PhysRevLett.77.198, gotteseman_1996} in Hilbert spaces composed of many qudits. While this framework is most known for its utility in designing quantum error correction codes such as the five-qubit \cite{gottesman1997stabilizercodesquantumerror,PhysRevA.54.3824,PhysRevLett.77.198} and Kitaev toric \cite{KITAEV20032} codes, it also provides a very efficient and convenient-to-handle representation of a broad class of multipartite states (both pure and mixed) composed of many qudits. In fact, the stabilizer formalism encompasses the well-known graph states  \cite{Hein1,Hein2} corresponding to one-dimensional stabilizer subspaces, but also classes of genuinely entangled mixed states defined on certain multidimensional stabilizer subspaces \cite{Makuta2023fullynonpositive}. 

In this work, we build upon the results of Refs. \cite{OtfriedStabilizer,PhysRevLett.94.060501}, which introduced entanglement witnesses for multiqubit graph states. Specifically, we generalize these results to higher-dimensional subspaces arising from the multiqudit stabilizer formalism. As a particular case of our construction, we also recover entanglement witnesses tailored to graph states of arbitrary prime local dimension. Robustness to white noise of the obtained EW's is then studied. Interestingly, we observe that the robustness of a given witness tailored to a particular graph state grows with the local dimension $d$, and also that, entanglement witnesses for genuinely entangled subspaces are in general superior in terms of robustness to white noise over witnesses tailored to states contained in those subspaces, which is a consequence of the fact that EW's for subspaces involve less stabilizing operators as compared to states. We finally provide a construction of a witness detecting GME tailored to a three-qubit two-dimensional subspaces, which does not originate from the stabilizer formalism, for which the stabilizing operators are non-local.

\section{Preliminaries}

In this section, we provide basic information about stabilizer formalism and genuine multipartite entanglement.

\subsection{Genuine multipartite entanglement}
Let us consider an $N$-partite Hilbert space $\mathcal{H}^{\textrm{tot}} = \mathcal{H}_1\otimes...\otimes\mathcal{H}_N$, where we assume that  $\mathcal{H}_i=\mathbb{C}^d$ for any $i$. We call a bipartition $Q|\overline{Q}$ a partition of the set of parties $[N]=\{1,\ldots,N\}$ into two, disjoint, and nonempty sets such that $Q\cup\overline{Q}=[N]$. Naturally, a bipartition $Q|\overline{Q}$ translates into a splitting of the total Hilbert space $\mathcal{H}^{\textrm{tot}}$ into Hilbert spaces corresponding to both sets, that is, $\mathcal{H}^{tot}=\mathcal{H}_{Q}\otimes\mathcal{H}_{\overline{Q}}$. 

The notion of a bipartition is central to the definition of \textit{genuine multipartite entanglement}
\cite{PhysRevLett.83.3562,Dur2_2000,Dur_2000,Acin_2001} (see also Ref. \cite{GUHNE20091} for a review): we say that a mixed state $\rho$ is biseparable if it admits the following decomposition.
\begin{equation}\label{eq: biseparable}
    \rho = \sum_{Q\subset[N]}p_Q\sum_iq_{i,Q}\rho_Q^i\otimes\rho^i_{\overline{Q}},
\end{equation}
for some $\rho_{Q}^{i}\in\mathcal{B}(\mathcal{H}_{Q})$, $\rho_{\overline{Q}}^{i}\in\mathcal{B}(\mathcal{H}_{{Q}})$, $\{q_{i,Q}\}_{i,Q}$ and $\{p_{Q}\}_{Q}$, where $\{q_{i,Q}\}_{i,Q}$ and $\{p_{Q}\}_{Q}$ are probability distributions. We then call a state $\rho$ genuinely multipartite entangled if it is not biseparable, that is, it does not admit the decomposition \eqref{eq: biseparable}.

Note that in the case of pure states, the above definition simplifies to the following: a state $\ket{\psi}$ is GME if for any bipartition $Q|\overline{Q}$ there are no 
pure states defined on $Q$ and $\overline{Q}$ such that
\begin{equation}
    \ket{\psi} = \ket{\psi_Q}\otimes\ket{\psi_{\overline{Q}}}.
\end{equation}
In other words, a pure GME state is one that cannot be written as a tensor product of any two other pure states. Lastly, the definition of pure GME states can directly be extended to entire subspaces of $\mathcal{H}^{\mathrm{tot}}$: we say that a subspace $V$ is GME if every state $\ket{\psi}\in V$ is GME. 

An example of a pure GME state is the well-known $N$-qubit $W$ state
\begin{equation}
    |W\rangle=\frac{1}{\sqrt{N}}(|10\ldots0\rangle+|01\ldots0\rangle+\ldots+|00\ldots1\rangle),
\end{equation}
which is a symmetric combination of all kets in which one qubit is 
in the excited state whereas the remaining $N-1$ qubits are in the ground state. Then, an example of a GME subspace would be one spanned by the $W$ state and the "negated" $W$ state 
\begin{equation}
    |\overline{W}\rangle=\frac{1}{\sqrt{N}}(|01\ldots1\rangle+|10\ldots1\rangle+\ldots+|11\ldots0\rangle),
\end{equation}
which is obtained from $W$ by flipping all bits $0\leftrightarrow 1$.

\subsection{Qudit stabilizer formalism}

In this work, we focus on quantum states---both pure and mixed---and subspaces that arise from the stabilizer formalism. Accordingly, we begin by introducing the fundamental concepts of this formalism that will be used throughout the paper. Let us consider a qudit Hilbert space $\mathbb{C}^{d}$, with the requirement that $d$ is prime.
Then, let us consider the Weyl-Heisenberg matrices $\omega^iX^jZ^k$, where $i,j,k\in\{0,\ldots,d-1\}$ and $X,Z\in\mathcal{B}(\mathbb{C}^d)$ are $d$-dimensional generalizations of the Pauli matrices
and are given by
\begin{equation}
\DIFdelbegin %DIFDELCMD < \begin{aligned}
%DIFDELCMD <     &X = \sum_{i=0}^d\ket{i+1}\bra{i},\qquad Z = \sum_{i=0}^d\omega^i\ket{i}\bra{i},
%DIFDELCMD < \end{aligned}%%%
\DIFdelend \DIFaddbegin \begin{aligned}
    &X = \sum_{j=0}^{d-1}\ket{j+1}\bra{j},\qquad Z = \sum_{j=0}^{d-1}\omega^j\ket{j}\bra{j},
\end{aligned}\DIFaddend 
\end{equation}
where the addition is done modulo $d$, and $\omega = \operatorname{exp}(2\pi \mathbb{i}/d)$. Importantly, they satisfy the following conditions
\begin{equation}
\begin{aligned}
    &X^d = \1, \quad Z^d = \1,\\
    &XX^\dagger = \1, \quad ZZ^\dagger = \1,
\end{aligned}
\end{equation}
i.e., they are unitary and their eigenvalues are $d$th roots of the unity. Moreover, they also satisfy the following commutation relation,
\begin{equation}
    ZX = \omega XZ.
\end{equation}

Now, the set of all $N$-fold tensor products of the Weyl-Heisenberg matrices endowed with the matrix multiplication forms the so-called generalized Pauli group $\mathbb{P}_{N,d}$. An abelian subgroup of $\mathbb{P}_{N,d}$ is called a stabilizer $\mathbb{S}\subset \mathbb{P}_{N,d}$ if $a\1\in\mathbb{S}$ implies that $a = 1$. For each such $\mathbb{S}$, one can identify a subspace $V_{\mathbb{S}}\subset\mathcal{H}^{\mathrm{tot}}$ which is stabilized by all elements of $\mathbb{S}$ in the following sense. 

\begin{equation}
\forall_{S\in \mathbb{S}} \quad S\ket{\psi}=\ket{\psi} \quad \Leftrightarrow \quad \ket{\psi} \in V_{\mathbb{S}}.
\end{equation}

As with other groups, it is convenient to represent a stabilizer $\mathbb{S}$ with its generators $\mathbb{S} = \langle G_1,\ldots,G_k\rangle $, i.e., a minimal set of elements such that any other element of a group can be represented by their product. 

One of the use-cases for the generators is the computation of the dimension of $V_{\mathbb{S}}$. Given that $d$ is prime and the number of generators equals $k$, we have that $\dim \Vs = d^{N-k}$ \cite{GHEORGHIU2014505}.

\subsection{Graph states}
An important subclass of stabilizer subspaces is the one-dimensional ones, known as stabilizer states. Among them, graph states form a notable subset, characterized by additional constraints derived from a graph $\mathcal{G}$. To formally define graph states, let us first introduce the concept of graphs.

A graph $\mathcal{G} = (V,E)$ is an ordered pair, where $V$ is a set of vertices, and $E$ is a set of edges, i.e., two-element, unordered sets of vertices.  If we allow multiple copies of the same edge in $E$ then such a graph is called a multigraph. While in this paper we will mostly focus on multigraphs, for simplicity, we will refer to them as graphs.

With each graph, we can associate an adjacency matrix $\Gamma$ whose entry $\Gamma_{i,j}$ corresponds to the number of edges connecting vertices $i$ and $j$. On Fig. \ref{fig:graphs}, we provided some examples of graphs.

\begin{figure}[H]
    \centering
    \subfigure[]{\includegraphics[width=0.3\linewidth]{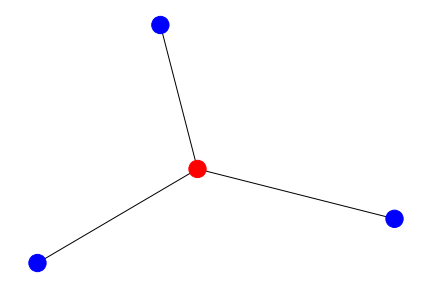}}
    \hfill
    \subfigure[]{\includegraphics[width=0.3\linewidth]{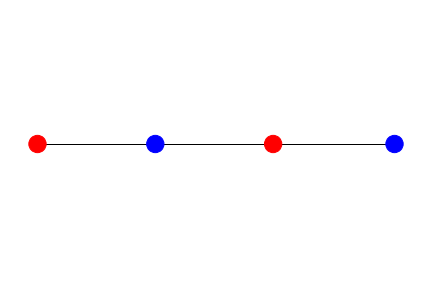}}
    \hfill
    \subfigure[]{ \includegraphics[width=0.3\linewidth]{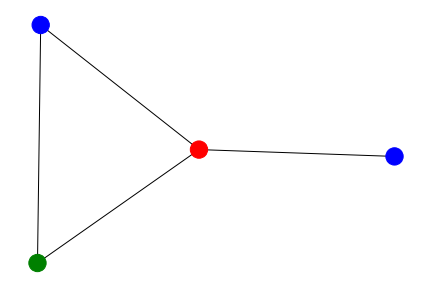}}
    \caption{\textbf{Exemplary graphs leading to four-partite graph states.} Graph (a) represents the GHZ state, (b) the cluster state, and (c) is an example of a graph with chromatic number $K=3$.}
    \label{fig:graphs}
\end{figure}

Let us finally introduce the concept of graph colorability and the related notion of chromatic number. A $k$-coloring of a graph $\mathcal{G}$ is a division of the vertices $V$ into $k$ disjointed subsets $\{C_{i}\}_{i=1}^{k}$. A proper coloring of $\mathcal{G}$ is a coloring for which $v_{1},v_{2}\in C_{i}$ only if $\{v_{1},v_{2}\}\notin E$. We say that the graph $\mathcal{G}$ is $k$-colorable if there exists a proper $k$-coloring of $\mathcal{G}$. Then, the chromatic number $K$ of a graph $\mathcal{G}$ is the smallest $k$ for which we can find a proper $k$-coloring of $\mathcal{G}$. 

Finally, let us define the graph states. Consider a graph $\mathcal{G} = (V,E)$ with an adjacency matrix $\Gamma$. A graph state $\ket{\mathcal{G}}$ is a unique state stabilized by the following generators
\begin{equation}\label{eq: generators}
    G_i = X_{i}\prod_{j=1}^{N}Z^{\Gamma_{i,j}}_{j},
\end{equation}
where $X_{j},Z_{j}$ are operators that act on the qudit $j$ with operators $X$ and $Z$ respectively, while acting with $\mathbb{1}$ on every other qudit. 

It is worth noting that any graph state can be expressed 
in the computational basis as \cite{Raissi_2024}:
\begin{equation}\label{eq: graph_state}
    \DIFdelbegin %DIFDELCMD < \ket{G} %%%
\DIFdelend \DIFaddbegin \ket{\mathcal G} \DIFaddend = \frac{1}{\sqrt{d^{N}}}\sum_{\vec{j}}\omega^{\tau(\vec{j})}\ket{\vec{j}},
\end{equation}
where $\vec{j}\in (\mathbb{Z}_{d})^{\otimes N}$ with $\mathbb{Z}_d=\{0,\ldots,d-1\}$, the sum is over all possible strings in $(\mathbb{Z}_{d})^{\otimes N}$, and the function $\tau(\cdot): (\mathbb{Z}_{d})^{\otimes N}\rightarrow \mathbb{Z}_{d}$ is given by
\begin{equation}\label{eq: tau}
\tau(\vec{j}) = \sum_{r,s=1}^{N} \mu_{r,s} j_{r}j_{s},
\end{equation}
where $\mu_{r,s}$ equals the number of times the edge $\{r,s\}$ appears in $E$.

In this paper, we also consider particular stabilizer states: Greenberger---Horne---Zeilinger (GHZ) state \cite{Greenberger1989} and the linear cluster state \cite{Briegel_2001,PhysRevLett.86.5188}. The generators of these states are given, respectively, by
\begin{equation}\label{eq: generators ghz}
\begin{aligned}
    &G_1^{GHZ} = \prod_{i=1}^{N} X_{i},\\
    &G_k^{GHZ} = Z_{k-1} Z^{\dagger}_{k}, \quad k=2,3,\ldots,N,
\end{aligned}
\end{equation}
and
\begin{equation}\label{eq: generators cluster}
\begin{aligned}
    &G_1^{C} = X_{1}Z_{2},\\
    &G_k^{C} = Z_{k-1}X_{k}Z_{k+1}, \quad k=2,3,\ldots,N-1,\\
    &G_N^{C} = Z_{N-1}X_{N}.
\end{aligned}
\end{equation}

Let us finally notice that the graph states are all genuinely multipartite entangled \cite{Hein1}

\subsection{Entanglement witnesses}

In order to detect genuine multipartite entanglement of a given state, one can employ a plethora of methods. Among them, one of the most commonly utilized and, at the same time, the most powerful, is that based on the concept of \textit{entanglement witness}. A witness $\W\in\mathcal{B}(\mathcal{H}^{\textrm{tot}})$ designed to detect GME is a Hermitian operator acting on $\mathcal{H}^{\mathrm{tot}}$ whose expectation value
on biseparable states is non-negative, 
\begin{equation}\label{EWcondition}
    \forall_{\rho_{\mathrm{bisep}}}\,\,\Tr(\mathcal{W}\rho_{\mathrm{bisep}})\geq 0,
\end{equation}
and there exists a GME state $\sigma$ for which
$\Tr(\mathcal{W}\sigma)<0$. In other words, $\TrM{\W\rho}<0$ constitutes a sufficient condition for genuine multipartite entanglement of $\rho$.

%that satisfies the following condition,
%
%\begin{equation}\label{eq: witness condtitions}
%    \TrM{\W\rho} =
%    \begin{cases}
%        & \geq 0 \quad \text{for all biseparable states},\\        
%        & <0 \quad \text{for some GME states}.
%    \end{cases}
%\end{equation}
%

\subsection{Measures for entanglement witnesses}

To assess the viability of a given entanglement witness for experimental implementations, one needs to consider its robustness to noise. To this end, we consider a state $\rho(p_{noise})$ defined as
\begin{equation}
    \rho(p_{\mathrm{noise}}) = p_{\mathrm{noise}}\frac{\1}{d^{N}} + (1 - p_{\mathrm{noise}})\ket{\psi}\bra{\psi},
\end{equation}
which is a probabilistic mixture of a target state $\ket{\psi}$ with white noise. Then, the noise robustness of a witness can be assessed by analyzing the amount of white noise we can add to a target state that still leads to a detection of entanglement. To this end, let us recall a useful notion of threshold probability. It is an upper-bound on the probability $p$ for which the witness $\mathcal{W}$ detects the entanglement of $\rho(p)$, i.e.,
\begin{equation}
    \operatorname{Tr}[W\rho(p_{\mathrm{limit}})]=0.
\end{equation}
From simple algebra, we obtain
\begin{equation}\label{eq: threshold probability}
    p_{\mathrm{limit}} = \frac{-\bra{\psi}\W\ket{\psi}}{d^{-N}\Tr(\W) - \bra{\psi}\W\ket{\psi}}.
\end{equation}

Another metric that one can use to assess the experimental viability of a given witness is the complexity of its experimental realization. For this task, usually the concept of local measurement settings (LMSs) is used, but first, we need to introduce a notion of local commutation \cite{OtfriedStabilizer}. Consider the following operators
\begin{equation}
    K = K^{(1)}\otimes \ldots \otimes K^{(N)}, \quad L = L^{(1)}\otimes \ldots \otimes L^{(N)}.
\end{equation}
We say that $K$ and $L$ commute locally if $[K^{(n)},L^{(n)}] = 0$ for every $n = \{1, \ldots, N\}$. Returning to LMSs, we simply say that 2 operators are measurable in the same LMS if they commute locally. Understandably, the metric in question is the number of LMSs needed to measure a given witness $\mathcal{W}$. This notion captures the intuition that in an experimental setup, where we can perform measurements on each party separately, on the same site, we can measure only operators that commute, so the operators for a whole system have to commute locally. 

As was observed in Ref. \cite{OtfriedStabilizer}, the number of LMSs for entanglement witness tailored to a graph state $\ket{\mathcal{G}}$ has a very elegant interpretation: it corresponds to the chromatic number $K$ of the graph $\mathcal{G}$.

The idea of LMSs is central to the construction of entanglement witnesses in Ref. \cite{OtfriedStabilizer}. For example, the entanglement witness tailored to the $N$-qubit GHZ state is given by
\begin{equation}
\mathcal{W}^{(GHZ)} \coloneqq 3\mathbb{1} - 2\left(P_{1}^{GHZ} + \prod_{i=2}^{N}P_{i}^{GHZ} \right),
\end{equation}
where $P_{i}=(\mathbb{1}+G_{i}^{GHZ})/2$ is the projector onto the eigenspace of $G_{i}^{GHZ}$ associated to the eigenvalue $+1$. Notice that this witness can be measured using only two LMSs: one for $P_{1}^{GHZ}$ and another for $\prod_{i=2}^{N}P_{i}^{GHZ}$. In this way, one maximizes the number of GME states detected by this witness while minimizing the number of LMSs. 

\section{Witnesses of GME for Graph States}
The main goal of this paper is to generalize the results of Ref. \cite{OtfriedStabilizer} in two ways: first, we construct GME witnesses for graph states of any prime local dimension, and, second, we expand the construction of EWs to stabilizer subspaces of dimension higher than one.

As a first result of this work, we will construct a projector-based witness for a graph state $\ket{\mathcal{G}}$.
\begin{Theorem} \label{Theorem Projector}
Let us consider a graph state $\ket{\mathcal{G}}\in(\mathbb{C}^d)^{\otimes N}$, where $d$ is prime. The following operator
    \begin{equation}\label{eq: witness projector}
         \Wt = \frac{1}{d}\1 - \ket{\mathcal{G}}\!\bra{\mathcal{G}}
    \end{equation}
    is an entanglement witness detecting GME in the vicinity of $\ket{\mathcal{G}}$.
\end{Theorem}

\begin{proof}
To show that $\widetilde{\mathcal{W}}$ is a witness of GME, we need to show that the condition 
in Eq. \eqref{EWcondition} is satisfied. At the same time, it is not difficult to observe that $\bra{\mathcal{G}}\Wt \ket{\mathcal{G}}=-(d-1)/d$. 

To prove that \eqref{EWcondition} holds in this case, let us consider a bipartition $Q|\overline{Q}$.
It is direct to realize that it is enough to prove that 
\begin{equation}\label{TintoDeVerano}
\operatorname{Tr}\left(\Wt \rho_{Q}\otimes \rho_{\overline{Q}}\right)\geqslant 0
\end{equation}
for any pair of density matrices $\rho_{Q}$ and $\rho_{\overline{Q}}$ because by linearity of the trace, the inequality will also follow for an arbitrary state which is separable across the bipartition $Q|\overline{Q}$.
One then observes that the above inequality will follow if
\begin{equation}
 \max_{Q|\overline{Q}} \max_{\ket{\phi}\in \Phi_{Q|\overline{Q}}} |\langle \phi\ket{\mathcal{G}}|^{2}\leqslant \frac{1}{d},
\end{equation}
where $\Phi_{Q|\overline{Q}}$ is the set of product states across $Q|\overline{Q}$. Notably, this fact follows directly from  \cite[Corollary 1]{makuta2025frustrationgraphformalismqudit} which states that for all GME stabilizer subspaces $V_{\mathbb{S}}$, the generalized geometric measure of entanglement $E_{GGM}(\cdot)$ \cite{PhysRevA.76.042309} equals
\begin{equation}
E_{GGM}(V_{\mathbb{S}}) = 1 - \max_{Q|\overline{Q}} \max_{\ket{\phi}\in \Phi_{Q|\overline{Q}}} \bra{\phi} \Pi_{V_{\mathbb{S}}}\ket{\phi} = 1 - \frac{1}{d},
\end{equation}
where $\Pi_{V_{\mathbb{S}}}$ is the projector onto $V_{\mathbb{S}}$. Taking $\Pi_{V_{\mathbb{S}}}=\ket{\mathcal{G}}\!\bra{\mathcal{G}}$, one obtains (\ref{TintoDeVerano}).
To complete the proof, we note that since the equation holds for any bipartition, the condition is also fulfilled by the linearity of the trace.
\end{proof}

The above entanglement witnesses as well as those that we obtain below exhibit some symmetries that are worth mentioning here: vertex relabelling in the graph $\mathcal{G}$ corresponds to relabelling of the subsystems in the corresponding entanglement witness, and, local complementation 
in the graph corresponds to implementing a local unitary transformation composed of Clifford operations to the graph state, and thus to the corresponding entanglement witness.

Let us now consider the experimental viability of $\Wt$. First, after simple computations, one can show that the threshold probability $p_{\textrm{limit}}$ equals 
\begin{equation}
p_{\textrm{limit}} = \frac{d-1}{d}\frac{1}{1-d^{-N}}.
\end{equation}
Next, to determine how many LMSs are needed to experimentally implement $\Wt$, we use the fact that a projector of any stabilizer state can be expressed as follows
\begin{equation}
    \proj{G} = \prod_{j=1}^N\frac{1}{d}\sum_{i=0}^{d-1}G_j^i,
\end{equation}
where $G_k$ are generators of a stabilizer of a given $\ket{\mathcal{G}}$. With the aid of this decomposition, it was shown in Ref. \cite{G_hne_2007} that in the case of $d=2$, the number of LMSs needed to implement experimentally projector-based witnesses grows linearly with the number of qubits $N$. Thus, although this witness offers high robustness against noise, it is impractical for larger systems. This fact is what motivates the search for other witnesses that offer easier experimental implementation.  

Notice that the projector-based witness (\ref{eq: witness projector}) involves measuring all the stabilizing operators. However, to uniquely identify a graph state, it is enough to measure only $N$ generators of the corresponding stabilizer. Therefore, it should be possible to limit the number of LMSs required to implement the stabilizer witness.

Taking this idea to its limits, one can construct a witness that consists of only generators of a stabilizer for a given $\ket{\mathcal{G}}$.
\begin{Theorem} \label{Theorem gens}
Consider a graph state $\ket{\mathcal{G}}$ of prime local dimension $d$, and the corresponding stabilizer $\mathbb{S} = \langle G_1,\ldots G_N\rangle$. Then the following operator 
\begin{equation}
\begin{aligned}
        &\W= \biggl\{N - \frac{d-1}{d} \biggl[1-\cos \biggl( \frac{2\pi}{d}\biggr)\biggr]\biggr\} \1 - 
    \frac{1}{2}\sum_{k=1}^{N}(G_{k} + G_{k}^{\dagger})
\end{aligned}
\end{equation}
is a witness detecting GME in the vicinity of $\ket{\mathcal{G}}$.
\end{Theorem}

To formulate the proof of the above theorem, we will have to use Lemma from \cite{OtfriedStabilizer}.
\begin{Lemma}\label{lemma std}
If $\Wt$ is a GME witness and there exists $\alpha > 0$ such that the following operator inequality
\begin{equation}
    \W \geqslant \alpha\Wt \label{eq:lemma}
\end{equation}
is satisfied, then $\W$ has nonnegative expected values on all biseparable states.
\end{Lemma}
\begin{proof}
    Eq. (\ref{eq:lemma}) implies that $\TrM{\W\rho} \geqslant \alpha\, \Tr (\Wt\rho)$. In particular, since $\Tr({\Wt\rho}) \geqslant 0$ for all separable states and $\alpha > 0$, then also $\Tr{(\W\rho)}\geqslant 0$ for all separable states, what finishes the proof.
\end{proof}
Now, if $\W$ has a negative expected value on some entangled state, this implies that $\W$ is an EW. In particular, if $\W = \alpha\Wt$, the witness $\W$ would detect as many states as $\Wt$, but, in general, $\W$ detects fewer states.

With this lemma, we can now proceed to the proof of Theorem \ref{Theorem gens}.

\begin{proof}
The idea of the proof is as follows. We start with the operator 
\begin{equation}\label{eq: witness c}
\W = c\1 - \frac{1}{2}\sum_{k=1}^N\left ( G_k + G_k^{\dagger} \right )
\end{equation}
and, taking $\Wt$ to be projector-based witness from Theorem \ref{Theorem Projector}, we make use of Lemma \ref{lemma std}. Then, by construction, we show that for $\alpha = 1 - \cos\left(\frac{2\pi}{d}\right)$ and $c = \bigl\{N - \frac{d-1}{d} \bigl[1-\cos \bigl( \frac{2\pi}{d}\bigr)\bigr]\bigr\}$ $\W \geqslant \alpha\Wt$ .

Substituting Eq. \eqref{eq: witness projector} and \eqref{eq: witness c} to Lemma \ref{lemma std} we get
\begin{equation}
c\1 - \frac{1}{2}\sum_{k=1}^N\left ( G_k + G_k^{\dagger} \right ) \geqslant \alpha \left(\frac{1}{d}\1 - \ket{\mathcal{G}}\!\bra{\mathcal{G}}\right).
\end{equation}
The projector $\ket{\mathcal{G}}\!\bra{\mathcal{G}}$ can be decomposed as a sum of stabilizing operators. Since all of the stabilizing operators mutually commute, they share a common eigenbasis; hence, the above inequality can be reduced to estimating the lowest eigenvalue of $\W-\alpha\Wt$ and choosing such $\alpha > 0$ that the lowest eigenvalue is non-negative. A general expression for the eigenvalues of $\W-\alpha\Wt$ reads

\begin{equation}
(\W-\alpha\Wt) \ket{v_{\vec{q}}}=\lambda_{\vec{q}}\ket{v_{\vec{q}}},
\end{equation}
where $\vec{q}=\{q_{1},\ldots,q_{N}\}$ such that $q_k\in\{0,\ldots,d-1\}$ for all $k$, and

\begin{equation}
\lambda_{\vec{q}} = c - \frac{\alpha}{d} - \frac{1}{2}\sum_{k=1}^{N}\left(\omega^{q_{k}} + \omega^{-q_{k}}\right) + \alpha\prod_{k=1}^{N}\frac{1}{d}\sum_{j=0}^{d-1}\omega^{jq_{k}}.
\end{equation}
There are two cases when the eigenvalues are the lowest: when all $\omega^{q_{k}} = 1$ or when all $\omega^{q_{k}} = 1$ except one $\omega^{q_{t}} \neq 1$, which we denote by $\vec{q}_{0}$ and $\vec{q}_{1}$ respectively. If we make sure they are positive, the others will be as well. Those two cases lead to
\begin{eqnarray}
 \lambda_{\vec{q}_0}&=&c - \frac{\alpha}{d} - N + \alpha,\nonumber\\
    \lambda_{\vec{q}_1}&=&c - \frac{\alpha}{d} - N + 1 - \text{Re}({\omega^{q_{t}}}).
\end{eqnarray}
%
\iffalse
\DIFaddbegin
\begin{equation}
\begin{aligned}
&\DIFadd{(\W-\alpha\Wt) \ket{v_{\vec{q}_{0}}}=(c - \frac{\alpha}{d} - N + \alpha) \ket{v_{\vec{q}_{0}}}}  \\
&\DIFadd{(\W-\alpha\Wt) \ket{v_{\vec{q}_{1}}}=( c - \frac{\alpha}{d} - N + 1 - \text{Re}({\omega^{q_{t}}})) \ket{v_{\vec{q}_{1}} }.} \\
\end{aligned}\\
\end{equation}
\DIFaddend
\DIFdelbegin
\begin{equation}
\begin{aligned}
& \DIFdel{\lambda_{1} = c - \frac{\alpha}{d} - N + \alpha}  \\
& \DIFdel{\lambda_{2} = c - \frac{\alpha}{d} - N + 1 - \text{Re}({\omega^{q_{t}}}).}.
\end{aligned}\\
\end{equation}
\fi
The eigenvalue $\lambda_{\vec{q}_{1}}$ can be estimated further by noticing that $\text{Re}({\omega^{q_{t}}})$ is the highest for $q_{t} = 1$. We want both $\lambda_{\vec{q}_i} \geq 0$. Solving this gives us optimal $c = N - \left(1-\frac{1}{d}\right) \left(1-\cos \left( \frac{2\pi}{d} \right) \right)$ and $\alpha = 1 - \cos(\frac{2\pi}{d})$. Notice that $\forall_{d}$ $\alpha > 0$, and that $\Tr({\W\proj{\mathcal{G}})} < 0$. Thus, we proved that 
\begin{equation}
\W = \biggl\{N - \frac{d-1}{d} \biggl[1-\cos \biggl( \frac{2\pi}{d}\biggr)\biggr]\biggr\}\1 - \frac{1}{2}\sum_{k=1}^{N}(G_{k} + G_{k}^{\dagger}) 
\end{equation}
detects GME in the vicinity of $\ket{\mathcal{G}}$.
\end{proof}

This witness has the minimal number of LMS required to detect GME since it measures only generators of a stabilizer, but this number depends on the size of the target graph state. On the other hand, one can quickly check that $p_{\textrm{limit}} = 1/N$, so for larger $N$ any amount of noise will make this witness unable to detect GME.

\section{Optimal Witnesses for Graph States}
The witnesses discussed so far were either hard to implement experimentally or their robustness against noise decreased unboundedly as the size of the system grew. In this section, we provide witnesses with minimal LMSs required to detect GME and the highest robustness against noise that can be obtained by Lemma \ref{lemma std}. We prove it in general for a graph state and then consider examples such as GHZ and cluster states.

\begin{Theorem} \label{Theorem G}
Consider a graph state $\ket{\mathcal{G}}$ of prime local dimension $d$, corresponding stabilizer $\mathbb{S} = \langle G_1,\ldots,G_N\rangle$ and the associated graph $\mathcal{G} = (V,E)$ with chromatic number $K$ and the corresponding coloring $\{C_{i}\}_{i=1}^{K}$. Then, the following operator
\begin{equation}
\begin{aligned}
&\Wo{\mathcal{G}} = [(K-1)d + 1]\1 - d\sum_{i=1}^{K}\prod_{j\in C_{i}}\frac{1}{d}\sum_{n=0}^{d-1}G_j^n 
\end{aligned}
\end{equation}
detects GME in the vicinity of $\ket{\mathcal{G}}$. 
\end{Theorem}

\begin{proof}
For a graph state $\ket{\mathcal{G}}$ associated with the graph $\mathcal{G}$ whose chromatic number is $K$, we need at least $K$ LMSs to detect GME. Let us consider the common eigenbasis of generators $G_{i}$ of the stabilizer of $\ket{\mathcal{G}}$, which we denote by $\mathcal{E}:=\{\ket{\vec{g}_i}\}_{i=1}^{d^N}$ such that 
\begin{equation}
G_{j}\ket{\vec{g}_i}=\omega^{g_{j}}\ket{\vec{g}_i},
\end{equation}
where each eigenvector $\ket{\vec{g}_i}\in(\mathbb{C}^d)^{\otimes N}$ is parameterized by $N$ numbers $g_i\in\{0,\ldots,d-1\}$, that is, $\ket{\vec{g}_i}\equiv|g_1,\ldots,g_N\rangle$.
Note that this is not a computational basis: in fact, many of the states $\ket{\vec{g}_i}$ are entangled.

Next, let us consider the following entanglement witness
\begin{equation}
\begin{aligned}
&\Wo{\mathcal{G}} = c\1 - d\sum_{i=1}^{K}\prod_{j\in C_{i}}\frac{1}{d}\sum_{n=0}^{d-1}G_j^n,
\end{aligned}
\end{equation}
where $c\in \mathbb{R}$ is an arbitrary coefficient which we determine below. Notice that in the basis $\mathcal{E}$, this witness can be expressed as
\begin{equation}
\DIFdelbegin %DIFDELCMD < \begin{aligned}
%DIFDELCMD <  \Wo{\mathcal{G}} &= c\1 - d\sum_{i=1}^{K}\prod_{j\in C_{i}}\frac{1}{d}\sum_{n=0}^{d-1}\sum_{\vec{g}\in(\mathbb{C}^{d})^{\otimes N}} \omega^{n g_{j}}\ket{\vec{g}}\!\bra{\vec{g}}\\
%DIFDELCMD < &= c\1 - d\sum_{i=1}^{K}\sum_{\vec{g}\in(\mathbb{C}^{d})^{\otimes N}} \left(\prod_{j\in C_{i}}\delta_{g_{j},0}\right)\ket{\vec{g}}\!\bra{\vec{g}},
%DIFDELCMD < \end{aligned}%%%
\DIFdelend \DIFaddbegin \begin{aligned}
 \Wo{\mathcal{G}} &= c\1 - d\sum_{i=1}^{K}\prod_{j\in C_{i}}\frac{1}{d}\sum_{n=0}^{d-1}\sum_{k=1}^{d^N} \omega^{n g_{j}}\ket{\vec{g}_k}\!\bra{\vec{g}_k}\\
&= c\1 - d\sum_{i=1}^{K}\sum_{k=1}^{d^N} \left(\prod_{j\in C_{i}}\delta_{g_{j},0}\right)\ket{\vec{g}_k}\!\bra{\vec{g}_k},
\end{aligned}\DIFaddend 
\end{equation}
where $\delta_{g_{j},0}$ is the Kronecker delta. This form of $ \Wo{\mathcal{G}}$ allows us to easily compute its trace as well as the expectation value in the graph state $|\mathcal{G}\rangle$,
\begin{equation}
\begin{aligned}
\operatorname{Tr}( \Wo{\mathcal{G}})&=c d^{N} -d \sum_{i=1}^{K}d^{N-|C_{i}|},\\
\bra{\mathcal{G}}\Wo{\mathcal{G}}\ket{\mathcal{G}}&=c-d K.
\end{aligned}
\end{equation}
Now, using Eq. \eqref{eq: threshold probability} we can calculate the threshold probability as a function of $c$, giving
\begin{equation}
\p = \frac{d K-c}{dK  -d \sum_{i=1}^{K}d^{-|C_{i}|}}.
\end{equation}
Clearly, if we wish to increase the noise robustness of the witness, we need to maximize $\p$, and so we need to find the minimal $c$.

To this end, let us examine for which values of $c$, the operator $\Wo{\mathcal{G}}$ is a witness of genuine multipartite entanglement. For this purpose, we consider the following inequality  
\begin{equation}\label{Ineq}
\Wo{\mathcal{G}}-\alpha\Wt \geqslant 0,
\end{equation}
where $\Wt$ is the projector-based witness \eqref{eq: witness projector} and $\alpha$ is real-valued parameter and, apply Lemma \ref{lemma std}. The eigenvalues of the operator on the left-hand side of (\ref{Ineq})
are given by
\begin{equation}
\lambda_{\vec{g}}=c-\frac{\alpha}{d}+\alpha\delta_{\vec{g}_k,\vec{0}} -d\sum_{i=1}^{K} \prod_{j\in C_{i}}\delta_{g_{j},0},
\end{equation}
where $\vec{0}\in (\mathbb{C}^{d})^{\otimes N}$ is a zero vector. There are two ways in which we could achieve the minimal eigenvalue. First, by achieving the absolute minimum of the last term in the above sum, which is the case for $\vec{g}_k=\vec{0}$. However, this also means that we have a positive contribution from the term $\alpha\delta_{\vec{g}_k,\vec{0}}$. The second option is to minimize the last term, but under the condition $\alpha\delta_{\vec{g},\vec{0}}=0$. This is achieved if $g_{j}=0$ for all $C_{i}$ except for one. 
All together, these two conditions give us
\begin{equation}
\begin{aligned}
c-\frac{\alpha}{d}+\alpha -dK &\geqslant 0,\\
c-\frac{\alpha}{d} -d(K-1)&\geqslant 0.
\end{aligned}
\end{equation}
The minimal $c$ is achieved when both of these bounds coincide, which implies $\alpha=d$ and so the minimal $c$ equals
\begin{equation}
c= (K-1)d+1.
\end{equation}
It is clear that the operator $\Wo{\mathcal{G}}$ has a negative expected value on $\ket{\mathcal{G}}$, which ends the proof.
\end{proof}

Notice, that since $c=(K-1)d+1$, the threshold probability $\p$ of $\Wo{\mathcal{G}}$ equals
\begin{equation}
    \p = \frac{d-1}{d}\frac{1}{K  - \sum_{i=1}^{K}d^{-|C_{i}|}}.
\end{equation}

With these results, we can consider what state would allow for the most robust detection of GME in this construction. From the form of $\p$ it is clear that states that require more LMSs are also less robust against noise. Since the minimal number of LMSs required to detect any entanglement is $K=2$, then the optimal state $\ket{\mathcal{G}}$ for GME detection with $\Wo{\mathcal{G}}$ has to have $K=2$. Moreover, the maximal $\p$ for $K=2$ is achieved when $|C_{1}|=1$ and $|C_{2}|=N-1$, which is exactly the case for the star graph.

The graph state associated to a star graph is equivalent, up to local unitary transformation, to the GHZ state. The witness $\Wo{GHZ}$ tailored to the GHZ state is given by
\begin{equation}
    \Wo{GHZ} = (d+1)\1 - d\left( \frac{1}{d}\sum_{n=0}^{d-1}G_1^n + \prod_{i=2}^N\frac{1}{d}\sum_{n=0}^{d-1}G_i^n\right), \label{eq:EW_GHZ}
\end{equation}
where $G_{i}$ are the generators of the GHZ stabilizer \eqref{eq: generators ghz}. The threshold probability of this witness equals
\begin{equation}
    \p = \frac{d-1}{2d - 1 - d^{2-N}} .\label{eq:plim_GHZ}
\end{equation}

To compare those results with other states, we may consider the cluster state. The graph corresponding to the cluster state also has chromatic number $K=2$, but both colors have either the same number of generators or differ by one generator, for even and odd $N$, respectively.  In this case, the entanglement witness is given by
\begin{equation}
    \Wo{C} = (d+1)\1 - d\left (\prod_{\textrm{even } i}\frac{1}{d}\sum_{n=0}^{d-1}G_i^n + \prod_{\textrm{odd } i}\frac{1}{d}\sum_{n=0}^{d-1}G_i^n \right ), \label{eq:EW_C}
\end{equation}
where $G_{i}$ are generators of the cluster state stabilizer \eqref{eq: generators cluster}. As compared to the witness tailored to the GHZ state, this one is slightly less robust against white noise because the 
critical value of $p$ reads
\begin{equation}
    \p = 
    \begin{cases}
        &\displaystyle\frac{d-1}{2d - 2d^{-N/2+1}} \quad \text{for even $N$},\\
        &\displaystyle\frac{d-1}{2d - d^{-N/2+3/2} - d^{-N/2+1/2}} \quad \text{for odd $N$} \label{eq:plim_C}.
    \end{cases}
\end{equation}

On Fig. \ref{fig:GHZvsCluster}, we provide plots comparing $\p$ of those states in 4 different cases.
\begin{figure*}[t]
    \centering
    \begin{minipage}[t]{0.49\linewidth}
    \centering
        \includegraphics[width=1\linewidth]{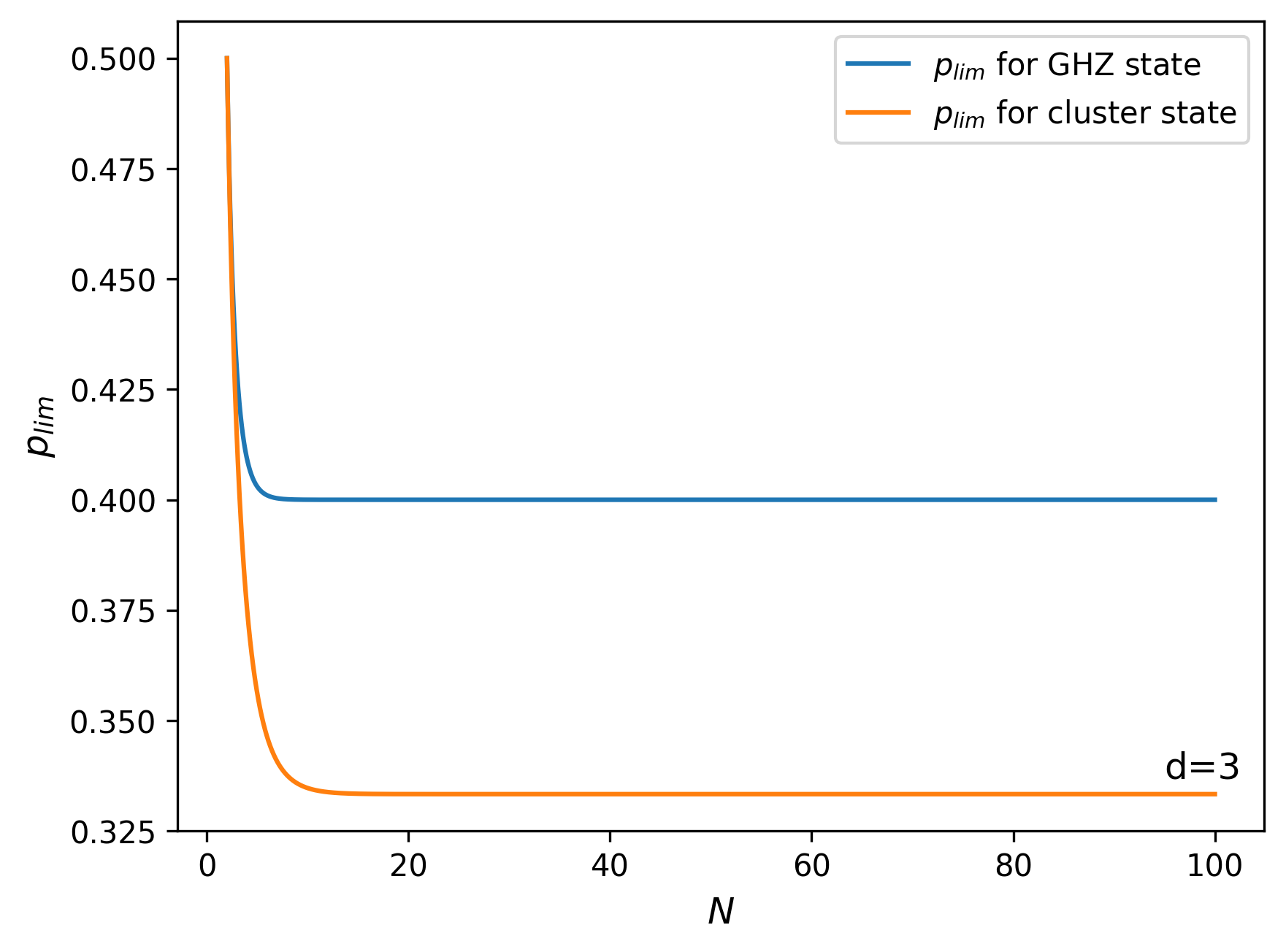}\\
        \includegraphics[width=1\linewidth]{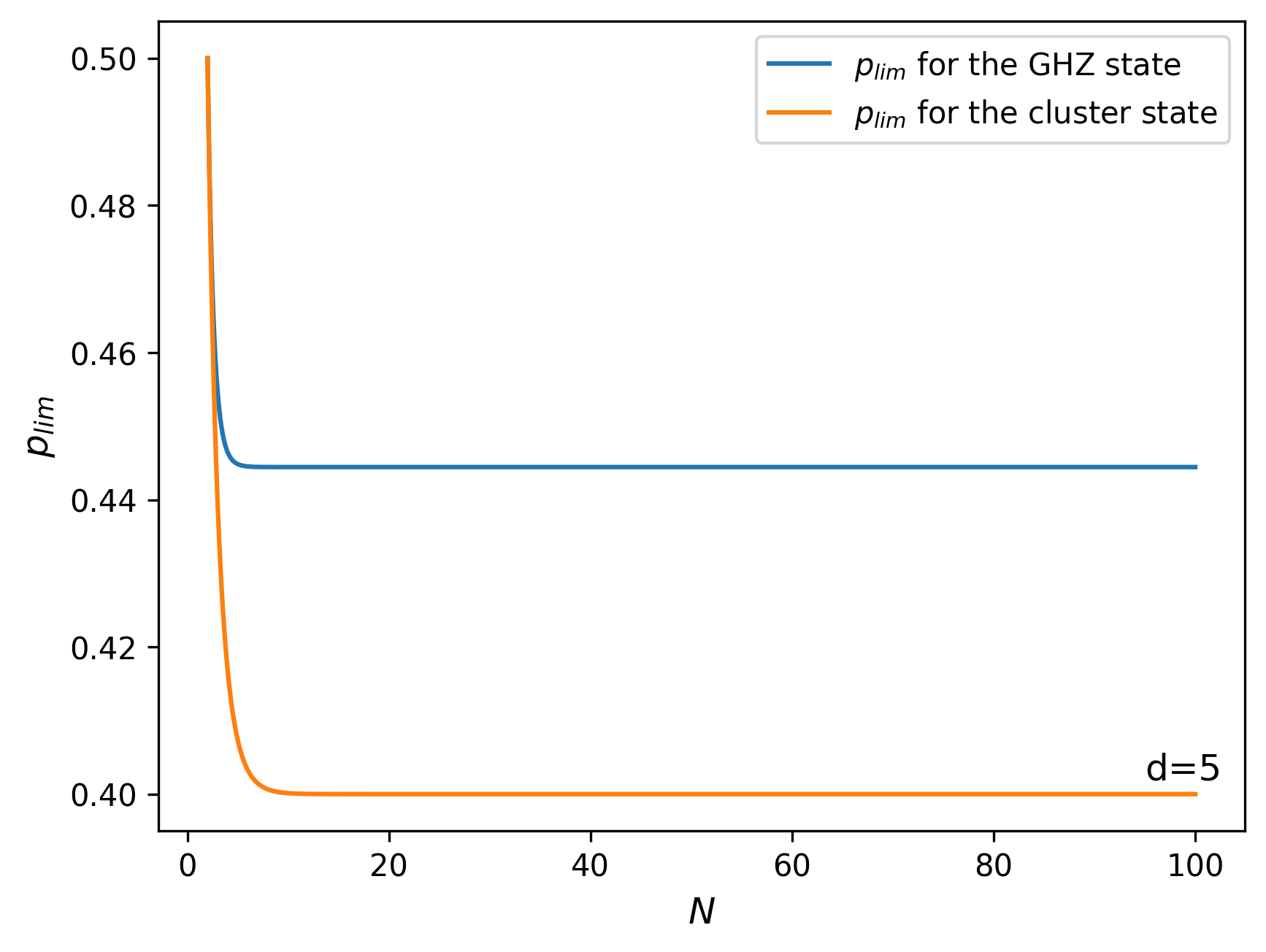}
    \end{minipage}
    \begin{minipage}[t]{0.49\linewidth}
    \centering
        \includegraphics[width=1\linewidth]{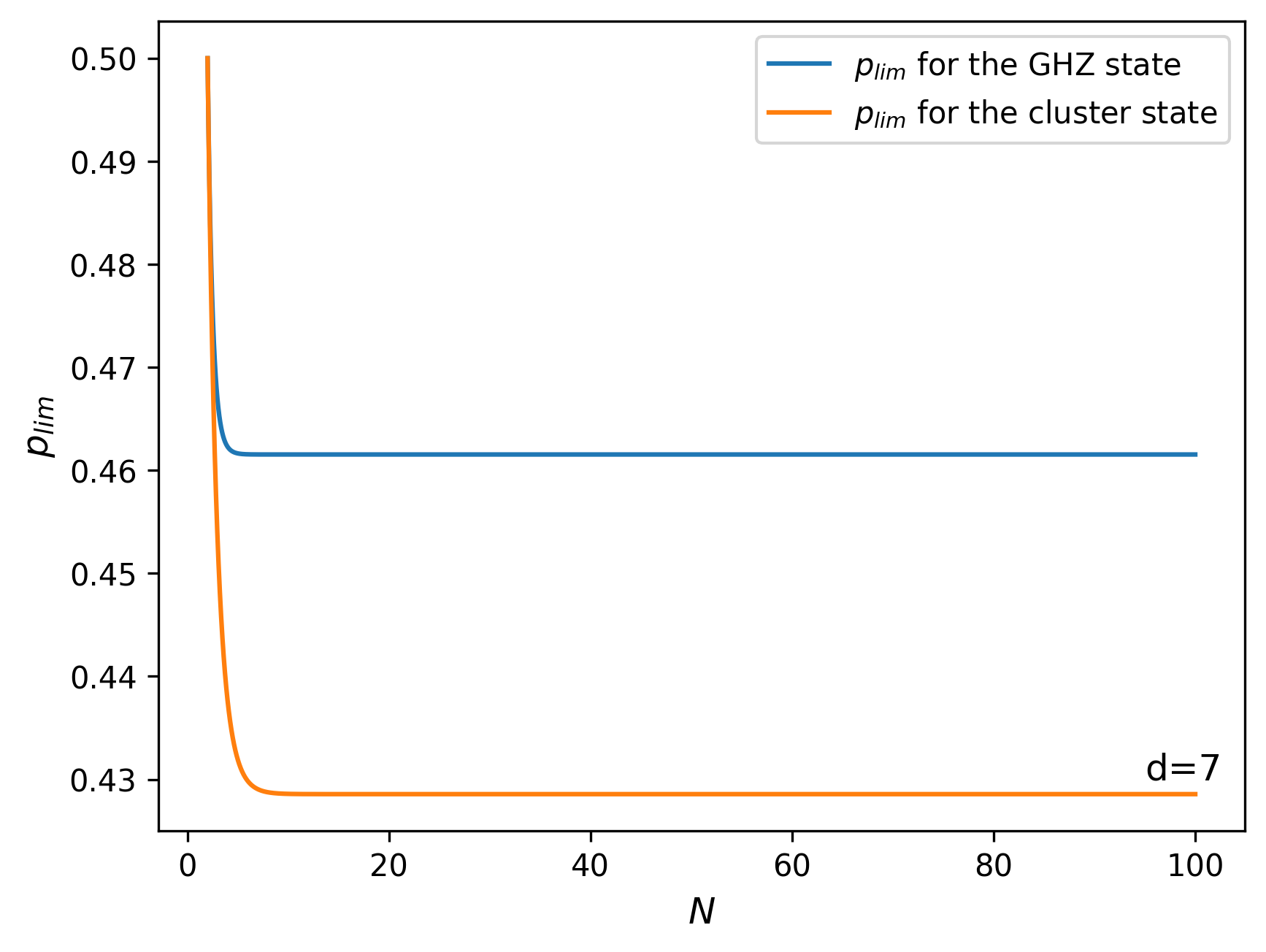}\\
        \includegraphics[width=1\linewidth]{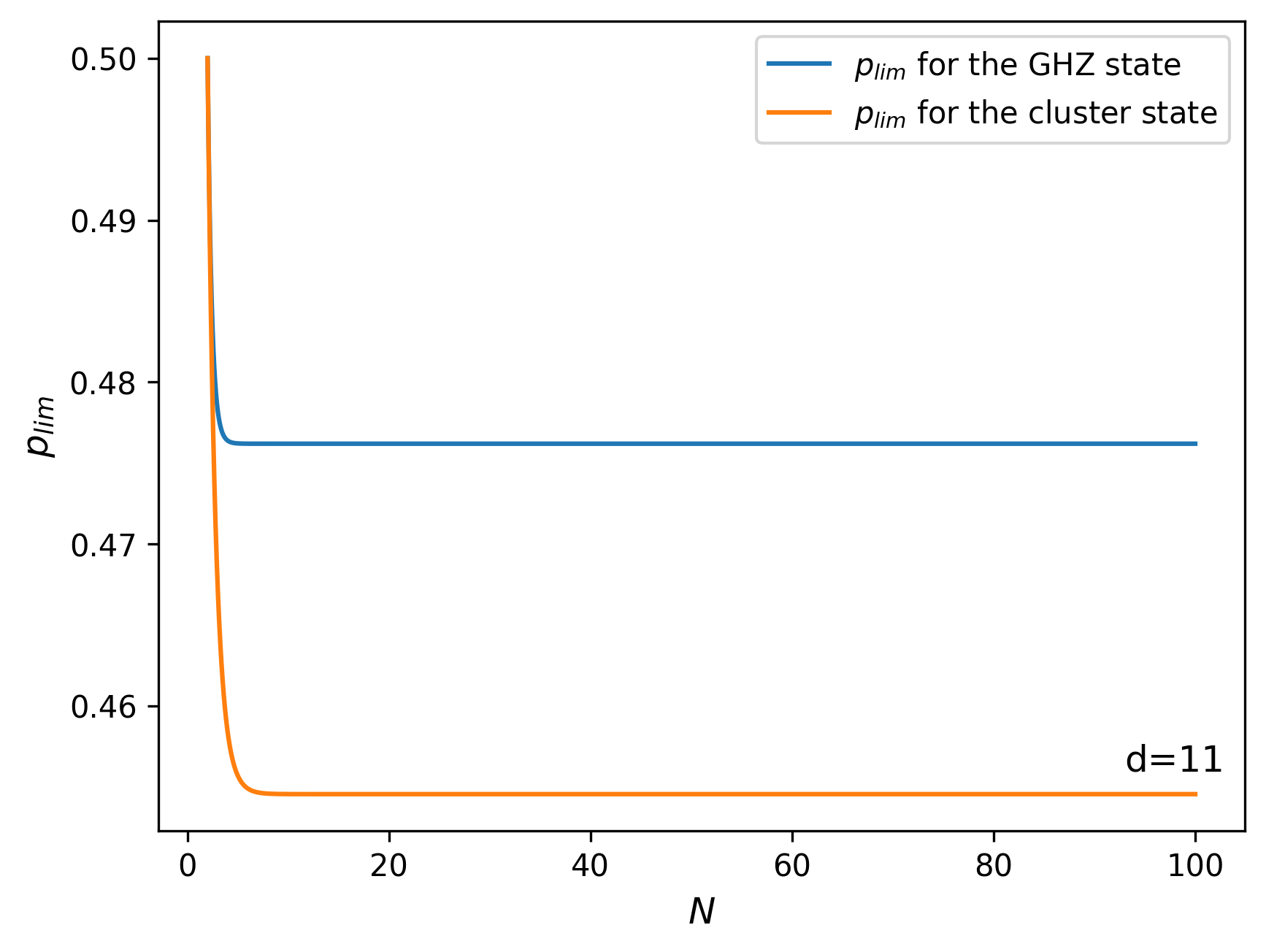}
    \end{minipage}
    \caption{\textbf{Comparison between the robustness to white noise of the GHZ and the cluster state}. The plots above show how $\p$ of EW for the GHZ state and the cluster state changes with $N$ for four different values of $d$. It can be observed that in each case, the EW for the GHZ state outperforms the EW for the cluster state. Moreover, for large enough $N$, both states asymptotically reach $p_{sat}$, different for each EW. Comparing all four plots, we see that, in both cases, $p_{sat}$ increases with $d$. For large $d$, $p_{sat}$ of both EW will approach $p_{sat} = 1/2$.}
    \label{fig:GHZvsCluster}
\end{figure*}

\section{Optimal Witnesses for GME stabilizer subspaces}

In this section, we first generalize the notions that we used for graph states to GME stabilizer subspaces. Moreover, we show a method to determine the minimal number of LMSs needed to detect GME for every state in GME stabilizer subspaces, analogous to the method for graph states. Then, we construct witnesses, including the most general results of this work - entanglement witnesses detecting GME for every state in the GME stabilizer subspace and its vicinity.

\subsection{Measures for entanglement witnesses for GME subspaces}
Before we move to the construction of the optimal EW for a GME stabilizer subspace, we first need to generalize the notion of critical probability $p_{\mathrm{limit}}$ introduced in the previous sections to subspaces. 

A natural generalization in the case of robustness against noise would be to consider the following noise model
\begin{equation}
    \rho(p_{\textrm{noise}}) = p_{\textrm{noise}}\frac{\1}{d^N} + (1-p_{\textrm{noise}})\rho_{\Vs},
\end{equation}
where $\rho_{\Vs}$ is an arbitrary state supported on $\Vs$. $\p$ corresponding to this noise model is 
\begin{equation}
    \p = \frac{-\TrM{\rho_{\Vs}\W}}{d^{-N}\TrM{\mathcal{W}}-\TrM{\rho_{\Vs}\W}}.
\end{equation}
Since, in general, such a value could in principle depend on the specific state, one can also define
\begin{equation}
    \p(\Vs) = \min_{\ket{\psi}\in \Vs}\p(\ket{\psi}).
\end{equation}
However, it is not difficult to observe that as long as witnesses are constructed solely from a sum of stabilizing operators, $\p$ is the same for any state $\rho_{\Vs}$, and so this issue does not manifest. Thus, in what follows, we consider a $\p$ for a particular state, namely the maximally mixed state in the subspace $\Vs$: $\rho_{\Vs} = 1/d^{N-k}P_{\Vs}$, where $P_{\Vs}$ is a projector onto the $\Vs$ and $k$ is the number of generators of the stabilizer $\mathbb{S}$. 

\subsection{Determining number of LMSs for GME subspaces}

We can also say something about determining LMSs for measuring stabilizer operators of a subspace. In the case of graph states, this problem was equivalent to finding the chromatic number of the corresponding graph. It turns out that a similar equivalence can be made in the case of a subspace.

First, we need to introduce the commutation matrices formalism, first proposed in Ref. \cite{Englbrecht2022transformationsof}. Consider a stabilizer $\mathbb{S} = \langle G_1,\ldots,G_k \rangle$. Recall that each $G_i = A_1\otimes\ldots\otimes A_N$, where $A_i$ are the Weyl-Heisenberg matrces that satisfy $A_iA_j - \omega^{\tau_{ij}}A_jA_i=0$, where $\tau_{ij}\in\{0,\ldots,d-1\}$ is determined by the exact form of $A_i$ and $A_j$. Now, with each site we associate a matrix $C^\alpha\in M_{k\times k}(\mathbb{Z}_{d})$ such that $(C^\alpha)_{ij} = \tau_{ij}$ of $A_i$ and $A_j$ at site $\alpha$ of $G_i$ and $G_j$ respectively. Let us then consider
\begin{equation}
    \tilde{C} = \sum_{\alpha=1}^NC^\alpha,
\end{equation}
where, importantly, $\tilde{C}\in M_{k\times k}(\mathbb{Z})$, i.e., the above sum is not modulo $d$. It is direct to see that if $(C)_{ij} = 0$, then $G_i$ and $G_j$ commute locally.

With matrix $\tilde{C}$ we can associate an adjacency matrix of the graph $\Gamma$ such that if $(\tilde{C})_{ij}\neq 0$ then $\Gamma_{i,j} = 1$ and 0 otherwise. By analogy to the graph state case, it is clear that the chromatic number of this graph will correspond to the minimal number of LMSs needed to measure all generators, and that vertices in the same color correspond to stabilizer operators that are measurable in the same LMS. Indeed, in the case of qubit graph states $\tilde{C}/2 = \Gamma$. Note, however, that this method has the same problem as in the case of graph states. Many graphs can be associated with the same stabilizer subspace and, in principle, their chromatic number $K$ may differ. So to find a proper lower bound on the number of LMSs needed to detect GME, we should minimize $K$ over all graphs that give rise to the same stabilizer subspace.

To illustrate this formalism let us present two simple examples. First, let us consider the three-qubit GHZ state. We want to show that, indeed, $\tilde{C}/2 = \Gamma$ holds in this case. 

The commutation matrices for each site for the GHZ state are given by
\begin{equation}
  C^1 = 
    \begin{bmatrix}
        0&1&0\\
        1&0&0\\
        0&0&0
    \end{bmatrix},
    \quad
    C^2=
    \begin{bmatrix}
        0&1&1\\
        1&0&0\\
        1&0&0
    \end{bmatrix},
    \quad
    C^3=
    \begin{bmatrix}
        0&0&1\\
        0&0&0\\
        1&0&0
    \end{bmatrix},
\end{equation}
and their sum amounts to,
\begin{equation}
    \tilde{C} = 
    \begin{bmatrix}
        0&2&2\\
        2&0&0\\
        2&0&0
    \end{bmatrix}
    =2\Gamma_{\mathrm{GHZ}}.
\end{equation}

The second example concerns a stabilizer subspace corresponding to the five-qubit code \cite{FiveQubit}. The generators of its stabilizer are given by
\begin{equation}
\begin{aligned}
G_1 = X_{1}Z_{2}Z_{3}X_{4},\qquad G_2 = X_{2}Z_{3}Z_{4}X_{5},\\
G_3=X_{1}X_{3}Z_{4}Z_{5},\qquad G_4 = Z_{1}X_{2}X_{4}Z_{5}.   
\end{aligned}
\end{equation}
In this case the matrix $\tilde{C}$ equals
\begin{equation}
    \tilde{C} = 
    \begin{bmatrix}
        0&2&2&2\\
        2&0&2&2\\
        2&2&0&2\\
        2&2&2&0\\
    \end{bmatrix}.
\end{equation}
Note that here it is also possible to associate a graph to this stabilizer by $\Gamma = \tilde{C}/2$, but that is not a general rule for all stabilizer subspaces.

\subsection{Entanglement witnesses for subspaces}

After considering the above, we can proceed to construct GME witnesses for GME stabilizer subspaces. As in the case of graph states, let us start with projector-based witnesses, which for a given stabilizer subspace $\Vs$ can be stated as
\begin{equation}
    \Wt_{\Vs} = \frac{1}{d}\1 - P_{\Vs},
\end{equation}
where $P_{\Vs}$ is the projector onto $\Vs$. The proof that this is a witness of GME follows the same lines as the proof of Theorem \ref{Theorem Projector}, and therefore we skip it here. 

Similarly to the projector-based witness for a graph state, this witness has 
\begin{equation}
    \p = \frac{d-1}{d}\frac{1}{1-d^{-N+k}},
\end{equation}
and the number of LMSs needed to detect GME with it grows linearly with the number of qudits $N$. 
In order to lower the number of LMSs, we could construct a witness as in Theorem \ref{Theorem gens}, but let us move to the following construction.
%
%As was the case of EWs tailored to graph state, here we are also concerned with minimizing %the number of LMSs needed for applying a given witness. This leads us to the following %construction.
%
\begin{Theorem} \label{Theorem Subspace}
Consider a GME stabilizer subspace on the Hilbert space of prime local dimension d with a corresponding stabilizer $\mathbb{S} = \langle G_1,\ldots, G_k\rangle$. The following operator

\begin{equation}
\begin{aligned}
&\Wo{\Vs} = [(K-1)d + 1]\1 - d\sum_{i=1}^{K}\prod_{j\in C_{i}}\frac{1}{d}\sum_{n=0}^{d-1}G_j^n
\end{aligned}
\end{equation}
detects GME around an arbitrary GME stabilizer subspace, where $K$ is the chromatic number of a graph associated with the subspace in a way described in the previous section, and $\{C_i\}_{i=1}^K$ is the corresponding coloring. Furthermore, this witness is the most robust against the noise among all witnesses $\mathcal{W}$ that need a minimal number of LMSs and fulfill $\W \geq \alpha\Wt$. 
\end{Theorem}

\begin{proof}
This proof follows the same idea as the proof of Theorem \ref{Theorem G}, and so here we only mention its major steps, while highlighting the modifications needed to adapt it for stabilizer subspaces.

First, we have to slightly adapt the definition of the basis $\mathcal{E}$ as the number of generators $k$ can now be lower than $N$, which implies that the relation
\begin{equation}
G_{j}\ket{\vec{g}_i}=\omega^{g_{j}}\ket{\vec{g}_i},
\end{equation}
does not identify states $\ket{\vec{g}_i}$ uniquely. To fix that, we simply take $\ket{g_{k+1},g_{k+2},\ldots,g_{N}}$ to be arbitrary orthonormal states.

With this, we can rewrite the witness
\begin{equation}
\begin{aligned}
&\Wo{\mathcal{V_{\mathbb{S}}}} = c\1 - d\sum_{i=1}^{K}\prod_{j\in C_{i}}\frac{1}{d}\sum_{n=0}^{d-1}G_j^n,
\end{aligned}
\end{equation}
as
\begin{equation}
\DIFdelbegin %DIFDELCMD < \begin{aligned}
%DIFDELCMD <  \Wo{V_{\mathbb{S}}} = c\1 - d\sum_{i=1}^{K}\sum_{\vec{g}\in(\mathbb{C}^{d})^{\otimes N}} \left(\prod_{j\in C_{i}}\delta_{g_{j},0}\right)\ket{\vec{g}}\!\bra{\vec{g}}.
%DIFDELCMD < \end{aligned}%%%
\DIFdelend \DIFaddbegin \begin{aligned}
 \Wo{V_{\mathbb{S}}} = c\1 - d\sum_{i=1}^{K}\sum\limits_{l=1}^{d^{N-k}} \left(\prod_{j\in C_{i}}\delta_{g_{j},0}\right)\proj{\vec{g}_l}.
\end{aligned}\DIFaddend 
\end{equation}
Then, the threshold probability equals
\begin{equation}
\p = \frac{d K-c}{dK  -d \sum_{i=1}^{K}d^{-|C_{i}|}}.
\end{equation}
Determining the minimal $c$ via Lemma \ref{lemma std} leads us to the non-negativity condition of the following eigenvalues
\begin{equation}
\lambda_{\vec{g}}=c-\frac{\alpha}{d}+\alpha\delta_{\vec{g}_l^k,\vec{0}} -d\sum_{i=1}^{K} \prod_{j\in C_{i}}\delta_{g_{j},0},
\end{equation}
where $\vec{g}_l^{k}$ is the vector consisting of the first $k$ entries of $\vec{g}_l$. From this condition, we derive
\begin{equation}
c= (K-1)d+1.
\end{equation}
It is clear that the operator $\Wo{\Vs}$ has a negative expected value on every state from $\Vs$, which ends the proof.
\end{proof}

It follows from the above that
\begin{equation}
    \p = \frac{d-1}{d}\frac{1}{K - \sum_{i=1}^K d^{-|C_{i}|}},\label{eq:p_lim_Vs}
\end{equation}
so it may seem that there is no advantage in detecting entanglement around a subspace than around a state. However, this is not necessarily the case since $\sum_{i=1}^{K}|C_{i}|=k\leqslant N$, and so there is a potential for an increase in $\p$ coming from the smaller number of generators of a stabilizer.

The simplest example in which this advantage can be witnessed is the stabilizer of the subspace of $(\mathbb{C}^{d})^{\otimes d}$ (i.e., $N=d$) generated by
\begin{equation}\label{eq: generators N=d}
G_{1} = \prod_{i=1}^{d}X_{i},\quad G_{2}=\prod_{i=1}^{d}Z_{i}.
\end{equation}
By \cite[Theorem 1]{Makuta2023fullynonpositive}, this subspace is GME. For $d=2$, this is a stabilizer of a Bell state, which could be also seen as a GHZ state for $N=2$, and so there is no advantage. However, for any $d>2$ the witness tailored to those subspaces produces a higher noise robustness than the one for GHZ, as for this subspace
\begin{equation}
    \p = \frac{1}{2}.
\end{equation}
Note that, as opposed to any witness for a graph state, the noise robustness of this witness does not scale with $N$ or $d$, and in fact, it achieves the maximal value of $\p$ possible in this witness construction. However, there is a big caveat to this claim as this witness only works for subspaces of $(\mathbb{C}^{d})^{\otimes d}$.

For the general case of $(\mathbb{C}^{d})^{\otimes N}$, this construction can be generalized to the  stabilizer generated by
\begin{equation}\label{eq:gens_opt}
\begin{aligned}
 G_{1}&=\prod_{i=1}^{N}X_{i},\\
 G_{j} &= \prod_{i=(d-1)(j-2)+1}^{(d-1)(j-1)+1}Z_{i} \quad\textrm{for }j\in\{2,\ldots,k-1\}\\
 G_{k}&=\left(\prod_{i=(d-1)(k-2)+1}^{N-1}Z_{i}\right)Z_{N}^{r},
\end{aligned}
\end{equation}
where $r=d-1 -(N-2\mod d-1) $ and $k=\lceil (N-1)/(d-1) \rceil$. Clearly, in this case, the number of required LMSs remains two, but now the number of qudits $N$ is independent of $d$. Moreover, for the case $N=d$, this stabilizer is equal to the one generated by \eqref{eq: generators N=d}. Interestingly, for $d=2$ this construction reproduces the GHZ tailored witness, implying that the advantage over graph state witnesses is achieved only for $d>2$. 

The threshold probability for the witness tailored to this stabilizer equals
\begin{equation}
\p = \frac{d-1}{2d -1-d^{-\lceil (N-1)/(d-1) \rceil+1}}.
\end{equation}
While we cannot claim that this is the highest $\p$ one can achieve with this witness construction, this is certainly a good candidate for such a maximum.

On Fig. \ref{fig:GHZvsOpt}, we provide plots comparing $\p$ of the EW for this subspace, with the EW for the GHZ state (\ref{eq:EW_GHZ}) in 4 different cases.
\begin{figure*}[t]
    \centering
    \begin{minipage}[t]{0.49\linewidth}
    \centering
        \includegraphics[width=1\linewidth]{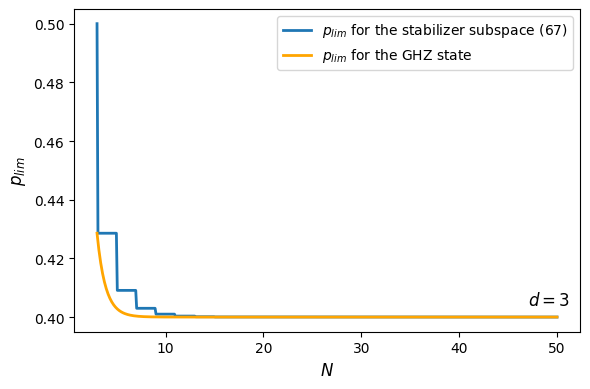}\\
        \includegraphics[width=1\linewidth]{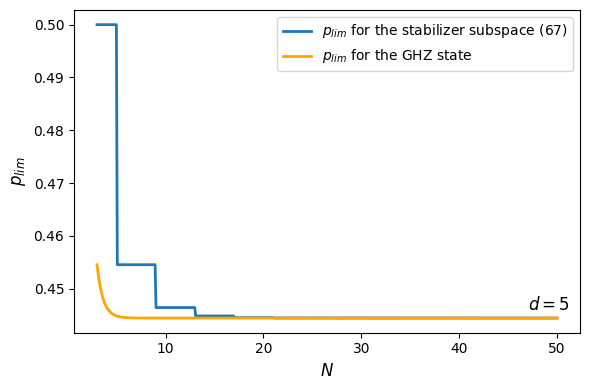}
    \end{minipage}
    \begin{minipage}[t]{0.49\linewidth}
    \centering
        \includegraphics[width=1\linewidth]{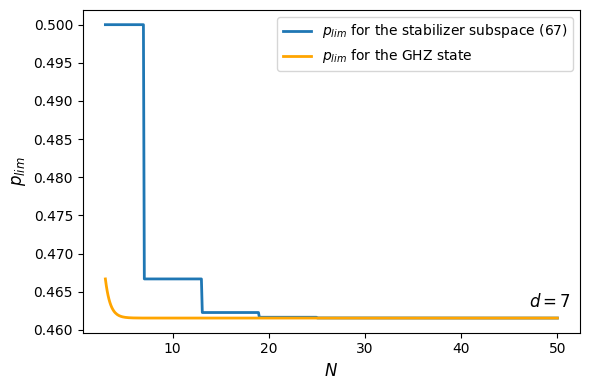}\\
        \includegraphics[width=1\linewidth]{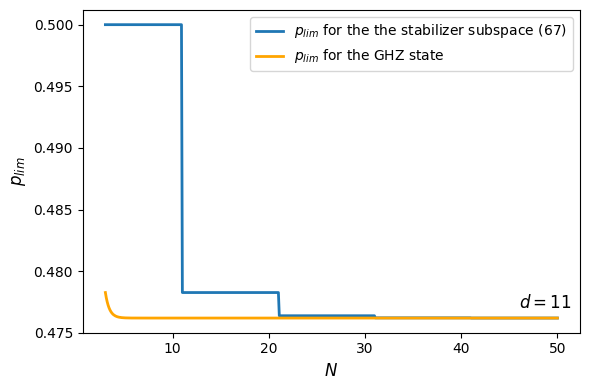}
    \end{minipage}
    \caption{\textbf{Comparison between the robustness to white noise of the GHZ and the stabilizer subspace (\ref{eq:gens_opt})}. The four plots above show how $\p$ changes with $N$ for different $d$ in the case of EWs for the GHZ state and the stabilizer subspace (\ref{eq:gens_opt}). It can be observed that for each $d$ the EW for the subspace performs better, but the results approach the same value $p_{sat}$ for large $N$. Moreover, $p_{sat}$ grows with $d$ and it will approach $p_{sat} = 1/2$ for large $d$.}
    \label{fig:GHZvsOpt}
\end{figure*}

\section{Entanglement witnesses beyond the stabilizer formalism}

In this section, we slightly deviate from the main object of our work - qudits. Here, we limit ourselves to qubits to ease the computations, but we do not see any problems in generalizing the results below to higher local dimensions. Our aim here is to further develop the approach of Ref. \cite{OtfriedStabilizer}, where it was shown that the construction of entanglement witnesses can be extended beyond the stabilizer formalism by considering non-local stabilizing operators. In the following, we define the \textit{ nonlocal stabilizer} of a subspace $V$, denoted $\mathbb{S}$, to be a group generated by $k$ generators constructed from the sum of tensor products of Weyl-Heisenberg/Pauli matrices, such that
\begin{equation}
    \forall_{\ket\psi\in V}\quad\forall_{S\in\mathbb{S}}\quad S\ket\psi=\ket\psi.
\end{equation}
It might be confusing to refer to this group as non-local \textit{stabilizer} since it is not a subgroup of the generalized Pauli group, but we decided to use this term because some of the properties of the stabilizer group relevant for our purpose carry over to the case of non-local operators.

To illustrate this method, let us consider the state $\ket{0\ldots0}$ whose stabilizer is $\mathbb{S}=\langle Z_1,\ldots, Z_N\rangle$, where $Z_i$ is the (generalized) $Z$ Pauli matrix acting on site $i$. Consider then a multipartite state $\ket{\psi}$ that does not originate from the stabilizer formalism. One of the approaches to construct the non-local stabilizer of this state is to find a unitary matrix $U$ such that
\begin{equation}
    \ket\psi = U\ket{0\ldots 0}. \label{eq:constab}
\end{equation}
Then, it follows that the non-local stabilizer of $\ket{\psi}$ is $\mathbb{S} = U\langle Z_1,\ldots ,Z_N\rangle U^{\dagger} =: \langle S_1,\ldots, S_N\rangle$. Note that the key point of this method is to allow non-local stabilizer operators, and the above is just one way to obtain them.

In the case of the $N$-qubit $\ket W$ state, we found a class of unitary matrices fulfilling (\ref{eq:constab}), with a simple decomposition into tensor products of Pauli matrices. Below, we present examples of these matrices in both the odd- and even-qubit case. In the first case, the matrix reads
\begin{equation}
\begin{aligned}
        &U^{W_{N=2j+1}} = \frac{1}{\sqrt{N}}
        (X\otimes \underbrace{Z\otimes\ldots\otimes Z}_{j}\otimes\underbrace{\1\otimes\ldots\otimes\1}_{j}+\\
        &+ \1\otimes X\otimes \underbrace{Z\otimes\ldots\otimes Z}_{j}\otimes\underbrace{\1\otimes\ldots\otimes\1}_{j-1} + \ldots +\\
        &+Z\otimes\ldots\otimes Z\otimes\1\otimes\ldots\otimes\1\otimes X).
\end{aligned}
\end{equation}
Notice that each next element in this sum can be obtained by applying the permutation
\begin{equation}
\begin{pmatrix} \label{eq:permutation}
    1&2&\ldots & N\\
    2&3&\ldots &1
\end{pmatrix}
\end{equation}
to the previous element. This, together with the fact that the first operator in the sum consists of $X$ on the first site, $(N-1)/2$ $Z$ matrices (equivalently, $N/2$ or $N/2-1$ in the case of even $N$. There are 2 choices because they both lead to $Z$ matrix appearing $N(N-1)/2$ times in the whole operator, which is one of the conditions in this class of unitary matrices and $\1$ at the rest of the sites, characterizes the class of unitaries we mentioned above.

For even $N$, the following unitary matrix fulfills 
(\ref{eq:constab}):
\begin{equation}\label{eq:U_2k}
\begin{aligned}
        &U^{W_{N=2j}} = \frac{1}{\sqrt{N}}[X\otimes \underbrace{Z\otimes\ldots\otimes Z}_{j-1}\otimes\underbrace{\1\otimes\ldots\otimes\1}_{j}(Z)+\\
        &+Z\otimes X\otimes\underbrace{ Z\otimes\ldots\otimes Z}_{j-1}\otimes\underbrace{\1\otimes\ldots\otimes\1}_{j-1}(\1)+\ldots+\\
    &+\underbrace{Z\otimes\ldots\otimes Z}_{j-2}\otimes\underbrace{\1\otimes\ldots\otimes\1}_{j}\otimes Z\otimes X(Z)].
\end{aligned}
\end{equation}
Notice that in each operator in the sum, there are $N+1$ matrices. The ($N+1$)th we put in parentheses to indicate its special role. To keep the property that the next operator in the sum can be constructed by applying the permutation (\ref{eq:permutation}) to the previous one, we have to do the following: we create a string of $N+1$ elements (so including the operator in parentheses), and, to obtain the next string, we use permutation (\ref{eq:permutation}). In this way, we construct $N$ strings, and the unitary matrix (\ref{eq:U_2k}) is constructed from the sum of the tensor product of the first $N$ operators from each string. 

In the case of an odd number of qubits $N$, we obtain the following generator of the $W$ state non-local stabilizer
\begin{eqnarray}
    &&S_{1}^{W_{N=2j+1}} = \frac{1}{N}[2YY\1\ldots\1 Z\1\ldots\1\nonumber+\\
    &&+ 2YZY\1\ldots \1 ZZZ\1\ldots\1\nonumber+\\
    &&+\ldots+2YZ\ldots ZYZ\ldots Z+2XZ\ldots ZXZ\ldots Z+\nonumber\\
   &&+2X\1Z\ldots Z\1XZ\ldots Z+\ldots +2X\1\ldots\1 Z\1\ldots\1 X +\nonumber\\ 
   &&+(N-2)Z\1\ldots\1].    
\end{eqnarray}
To obtain the next generator, we have to apply the permutation (\ref{eq:permutation}) to each term in the previous generator. 

With these operators, we can construct non-local stabilizer witnesses as in previous sections. For instance
\begin{Theorem}\label{Theorem W stab}
Consider the $N$-qubit $W$ state, where $N$ is an odd number, and its non-local stabilizer $\mathbb{S} = \langle S_1,\ldots, S_N\rangle$. The following operator
\begin{equation}
    \W^{W_N=2j+1} = \frac{N^2-1}{N}\1 - \sum_{k=1}^NS^{W_{N=2j+1}}_k
\end{equation}
detects GME around the $N$-qubit $W$ state.
\end{Theorem}
\begin{proof}
    The proof is the same as that of Theorem \ref{Theorem gens}. We first consider $\W^{W_{2j+1}} = c\1 - \sum_{k=1}^NS^{W_{2j+1}}_k$ and, we find such $c$ for which the inequality $\W^{W_{N=2j+1}} \geq \alpha\Wt^{W_N}$ holds for some positive $\alpha$ and
    \begin{equation}
        \Wt^{W_N} =\frac{N-1}{N}\1 - \proj{W^N}.
    \end{equation}
    The above was proven to be a projector-based EW in \cite{TothWState}. Then, using Lemma \ref{lemma std}, we finish the proof.

    To begin, the general expression for the eigenvalue of $A = \W^{W_N}-\Wt^{W_N}$ is 
    \begin{equation}
        \lambda = c - \alpha\frac{N-1}{N} - \sum_{k=1}^N(-1)^{q_k} + \alpha\prod_{k=1}^N\frac{1}{2}[1+(-1)^{q_k}].
    \end{equation}
    There are two distinct cases when $\lambda$ is the lowest. When $q_k=0$ for all $k=1,\ldots,N$, or when some $q_t=1$ and the rest of $q_k=0$ for $k=1,\ldots t-1,t+1\ldots, N$. Those cases give us
    \begin{equation}
        \lambda_1=c-N-\alpha\frac{N-1}{N} + \alpha, \quad \lambda_2 = c-N+1-\alpha\frac{N-1}{N}.
    \end{equation}
    Imposing that both $\lambda_1\geq 0$ and $\lambda_2\geq 0$, we obtain the optimal $c = (N^2-1)/N$. 
\end{proof}

Unfortunately, this EW does not give us any advantage over the projector-based witness for the $W$ state. It detects states mixed with white noise for any $p<1/N^2$, so it performs poorly even in small systems. Computing the number of LMSs required to use this witness is challenging, but for the 3- and 5-qubit examples, it proved less efficient than the projector-based witness. Since the operators are non-local, it no longer holds that the minimal number of LMSs required to detect GME is the number of LMSs required to measure all the generators of the stabilizer, e.g., by considering products of generators, we may obtain a witness that requires fewer LMSs. A perfect example is EW from \cite[Theorem 8]{OtfriedStabilizer}, where, although it has still lower $\p$ than the projector-based witness, it requires only three LMS, compared to 5 LMSs required by the projector-based EW. 

It is also possible to construct an EW for subspaces in this formalism. In the case of a subspace spanned by two states $\ket{\psi_1}$ and $\ket{\psi_2}$, we have to find a unitary matrix $U$ fulfilling $U\ket{000}=\ket{\psi_1}$ and $U\ket{001}=\ket{\psi_2}$. For example, consider a subspace $V$ spanned by the 3-qubit $\overline{W}$ state and 3-qubit $W$ state. Since we already have the non-local stabilizer of $\ket W$, we do not have to start from the beginning again. It can be checked that $V$ is stabilized by $\mathbb{S}=\langle S^{W_{3,2}}_1S^{W_{3,2}}_2,S^{W_{3,2}}_1S^{W_{3,2}}_3\rangle =: \langle S_{1}^{W,\overline{W}_{3,2}},S_2^{W,\overline{W}_{3,2}}\rangle$. With these operators, we can construct an EW similar to the one from \cite[Theorem 8]{OtfriedStabilizer}, but first, we need to obtain a projector-based witness for this subspace.

\begin{Theorem}
     Consider a subspace $V$ spanned by the 3-qubit $W$ state and the 3-qubit $\overline{W}$. The following operator
     \begin{equation}
         \Wt^{W,\overline{W}_{3,2}} = \frac{2}{3}\1 - P_V \label{eq:projb wit WWb}
     \end{equation}
     detects GME for every state in $V$ and its vicinity.
\end{Theorem}
\begin{proof}
    To prove that the above operator detects GME, we have to show that
    \begin{equation}
        c=\max_{Q|\overline{Q}}\max_{\ket\phi\in\mathcal{P}}\bra{\phi}P_V\ket{\phi} = \frac{2}{3}.
    \end{equation}
    In \cite{Demianowicz2, geomeasure} it was proven that the above is equivalent to
    \begin{equation}
       c=\min_{\ket{\psi}\in V}\max_{Q|\overline{Q}}\max_{\ket\phi\in\mathcal{P}}|\scalprod{\psi}{\phi}|^2=\frac{2}{3}.
    \end{equation}
    From \cite{SchmidtTrick} it follows that the first maximum in the above is equivalent to the square of the maximum Schmidt coefficient $\lambda_i$ of $\ket\psi$. Furthermore, both the GHZ and the $W$ state are symmetric under swapping of qubits; thus, we can limit ourselves to considering just one particular bipartition. After these remarks, this proof reduces to showing that 
    \begin{equation}
        c=\min_{\substack{\ket{\psi} = a\ket{\overline{W}} + b\ket{W}, \\ |a|^2 + |b|^2 = 1}}\max_{i}\lambda_i=\frac{2}{3}.
    \end{equation}
    This problem can be solved analytically. First, we compute the Schmidt coefficients of $\ket\psi$ and take the maximal one. Then, we minimize it over $a$ and $b$ with the constraint that $|a|^2+|b|^2=1$. Indeed, after these computations, one obtains $c=2/3$.
\end{proof}
With this result, we can construct a non-local stabilizer witness for the subspace $V$.
\begin{Theorem}
    Consider the subspace $V$ spanned by the 3-qubit $W$ and $\overline{W}$ states and the corresponding non-local stabilizer $\mathbb{S} = \langle S_1, S_2\rangle$. The following operator
    \begin{equation}
        \W^{W,\overline{W}_{3,2}} = \frac{5}{3}\1 - S_1^{W,\overline{W}_{3,2}}-S_2^{W,\overline{W}_{3,2}} \label{eq:stab_wit_WWb}
    \end{equation}
    detects GME in every state in $V$ and its vicinity. 
\end{Theorem}
\begin{proof}
    The proof is the same as those of Theorems \ref{Theorem gens} and \ref{Theorem W stab}. We first consider an operator 
    \begin{equation}
    \W = c\1 - S_1^{W,\overline{W}_{3,2}}-S_2^{W,\overline{W}_{3,2}}.    
    \end{equation}
     One then finds that for $c=5/2$ the inequality $\W\geq\alpha\Wt$ is fulfilled for some $\alpha>0$ and the projector-based witness $\Wt$ from (\ref{eq:projb wit WWb}). It finally follows from Lemma \ref{lemma std} that the EW (\ref{eq:stab_wit_WWb}) detects GME.
\end{proof}

Actually, the two just described EWs are almost the same since in the EW (\ref{eq:stab_wit_WWb}), we measure almost all non-local stabilizer operators. We can suspect that for a larger number of qubits $N$, the projector-based EW will offer higher robustness against noise, while the carefully constructed non-local stabilizer witness will require fewer LMSs. In this case, both EWs require 3 LMSs; the projector-based EW is characterized by $\p=8/21$ and the stabilizer EW by $\p = 1/6$. 

Note that finding a unitary matrix, which would rotate the state $\ket{000}$ to the state $\ket{\psi}$ around which we want to detect entanglement, with a simple decomposition into Pauli matrices, is generally a hard problem. In most cases, we could not find the non-local stabilizer operators with as simple decompositions as in the above cases. For example, we consider a subspace spanned by the 3-qubit GHZ state and the 3-qubit $W$ state. In this case, we only found unitary matrices similar to
\begin{eqnarray}
&&U^{GHZ,W_{3,2}}\nonumber\\
&&=
\begin{bmatrix}
0&1/\sqrt{2}&0&0&0&0&0&1/\sqrt{2}\\
1/\sqrt{3}&0&1/\sqrt{2}&0&1/\sqrt{6}&0&0&0\\
1/\sqrt{3}&0&-1/\sqrt{2}&0&1/\sqrt{6}&0&0&0\\
0&0&0&1&0&0&0&0\\
1/\sqrt{3}&0&0&0&-2/\sqrt{6}&0&0&0\\
0&0&0&0&0&1&0&0\\
0&0&0&0&0&0&1&0\\
0&1/\sqrt{2}&0&0&0&0&0&-1/\sqrt{2}
\end{bmatrix},\nonumber\\
\end{eqnarray}
which has all 64 coefficients in the decomposition into Pauli matrices. The non-local stabilizing operators that we obtain by rotating $Z^{(1)}$ and $Z^{(2)}$ are
\begin{equation}
S^{GHZ,W_{3,2}}_1=\left[
\begin{array}{cccccccc}
 0 & 0 & 0 & 0 & 0 & 0 & 0 & 1 \\
 0 & 2/3 & -1/3 & 0 & 2/3 & 0 & 0 & 0 \\
 0 & -1/3 & 2/3 & 0 & 2/3 & 0 & 0 & 0 \\
 0 & 0 & 0 & 1 & 0 & 0 & 0 & 0 \\
 0 & 2/3 & 2/3 & 0 & -1/3 & 0 & 0 & 0 \\
 0 & 0 & 0 & 0 & 0 & -1 & 0 & 0 \\
 0 & 0 & 0 & 0 & 0 & 0 & -1 & 0 \\
 1 & 0 & 0 & 0 & 0 & 0 & 0 & 0 \\
\end{array}
\right]
\end{equation}
and
\begin{equation}
S^{GHZ,W_{3,2}}_2=\left[
\begin{array}{cccccccc}
 0 & 0 & 0 & 0 & 0 & 0 & 0 & 1 \\
 0 & 0 & 1 & 0 & 0 & 0 & 0 & 0 \\
 0 & 1 & 0 & 0 & 0 & 0 & 0 & 0 \\
 0 & 0 & 0 & -1 & 0 & 0 & 0 & 0 \\
 0 & 0 & 0 & 0 & 1 & 0 & 0 & 0 \\
 0 & 0 & 0 & 0 & 0 & 1 & 0 & 0 \\
 0 & 0 & 0 & 0 & 0 & 0 & -1 & 0 \\
 1 & 0 & 0 & 0 & 0 & 0 & 0 & 0 \\
\end{array}
\right].
\end{equation}
Fortunately, they have a lot fewer non-zero coefficients in the decomposition into tensor products of Pauli matrices. To construct a non-local stabilizer EW for this subspace, we first have to obtain a projector-based witness. Following the same procedure as in the case of the subspace spanned by the $W$ state and the $\overline{W}$ state, we can prove that the operators
\begin{equation}
    \Wt^{GHZ,W_{3,2}} = \frac{1}{2}\1 -P_V
\end{equation}
and
\begin{equation}
    \W^{GHZ,W_{3,2}} = \frac{2}{3}\1 - S^{GHZ,W_{3,2}}_1-S^{GHZ,W_{3,2}}_2
\end{equation}
detect GME in every state from the subspace $V$ and its vicinity. Both of these operators have complex decomposition into tensor products of Pauli matrices and require a lot of LMS to be implemented. 

Our examples show that the problem with stabilizer EWs with non-local stabilizer operators is not that they are hard to construct. Indeed, we showed many examples that were not hard to obtain, and most of them are easy to generalize further. The problem is that the stabilizer witnesses were a promising alternative to projector-based EWs because they required fewer LMSs to implement, with worse, but still acceptable, noise robustness. In general, allowing for non-local stabilizer operators makes it harder to obtain an easy-to-implement EW. Yet, examples such as 3-qubit $W$ state suggest that there still might be a place for them, but it would require a method to find the unitary matrices with the simplest decomposition into tensor products of Pauli/Weyl-Heisenberg matrices. Once the optimal non-local stabilizer operators are found, a method to find an EW with the fewest LMS would be needed, since, in the case of non-local stabilizer operators, it no longer holds that the lowest number of LMSs required to detect GME is the number of LMSs required to measure all generators of the stabilizer.

\section{Conclusions}

In summary, we have generalized the genuinely multipartite entanglement witnesses with high noise robustness introduced in Ref. \cite{OtfriedStabilizer} for multiqubit graph states in two directions: (i) to systems with arbitrary prime local dimension $d$, and (ii) to stabilizer subspaces of dimension greater than one.
In the first case, we have constructed GME entanglement witnesses tailored to graph states of higher local dimension than two. We have also shown that the best robustness to white noise is achieved by considering a witness tailored to the GHZ state. 

For the latter, we showed how the construction can be extended to the cases in which the stabilizer does not specify a state, but rather a whole subspace of states. We have also shown that considering witnesses tailored to the stabilizer subspaces offers a clear advantage in noise robustness as compared to their graph state counterparts. In fact, for the case where the number of qudits equals the local dimension of a Hilbert space, we have formulated the most noise robust witness that can be created via this construction. For more general case of arbitrary local dimension and number of qudits, we have proposed a witness that offers a noise robustness higher than the one achieved by the graph state witness for any local dimension higher than two.

This work opens several avenues for further research. The most immediate is to identify, within this framework, entanglement witnesses that exhibit optimal noise robustness. While one candidate would be the stabilizer subspace and its corresponding witness which we proposed here, we have not provided a proof of optimality. Another promising line of research, explored already in Ref. \cite{OtfriedStabilizer} for some particular states, involves extending the method of constructing witnesses from tensor products of Pauli/Weyl-Heisenberg matrices to entangled quantum states beyond the stabilizer formalism, by considering sums of operators from the generalized Pauli group.

\section{Acknowledgments}

This project is supported by the National Science Centre (Poland) through
the SONATA BIS project No. 2019/34/E/ST2/00369 and has received funding from the European Union's Horizon Europe research and innovation programme under grant
agreement No 101080086 NeQST.

%\bibliographystyle{unsrt}
%\bibliography{references}

\end{document}